\documentclass[12pt]{article}

\usepackage{amssymb}
\usepackage{graphicx}
\usepackage{amsmath}
\usepackage{amsfonts,amsthm}

\newtheorem{theorem}{Theorem}
\newtheorem{corollary}{Corollary}

\newtheorem{definition}{Definition}
\newtheorem{proposition}{Proposition}

\usepackage[dvipdfmx]{color}
\newcommand{\bm}[1]{\mbox{\boldmath $#1$}}
\newcommand{\argmin}{\mathop{\rm argmin}}
\newcommand{\argmax}{\mathop{\rm argmax}}

\newcommand{\bms}[1]{\mbox{\scriptsize\boldmath $#1$}}

\title{Minimum information divergence of Q-functions \\ for dynamic treatment resumes
}

\author{Shinto Eguchi
\thanks{The Institute of Statistical Mathematics, Japan.  Email: {\tt eguchi@ism.ac.jp}}}

\date{\today}

\begin{document}
\bibliographystyle{jae.bst}
\maketitle

\begin{abstract}
This paper aims at presenting a new application of information geometry to reinforcement learning focusing on dynamic treatment resumes.
In a standard framework of reinforcement learning, a Q-function is defined as the conditional expectation of a reward given a state and an action for a single-stage situation. 
We introduce an equivalence relation, called the policy equivalence, in the space of all the Q-functions.
A class of information divergence is defined in the Q-function space for every stage.
The main objective is to propose an estimator of the optimal policy function by a method of  minimum information divergence based on a dataset of trajectories.
In particular, we discuss the $\gamma$-power divergence that is shown to have an advantageous property such that
the $\gamma$-power divergence between policy-equivalent Q-functions vanishes.
This property essentially works to seek the optimal policy, which is discussed in a framework of a semiparametric model for the Q-function.  
The specific choices of power index $\gamma$ give interesting relationships of the value function,
and the geometric and harmonic means of the Q-function.  
A numerical experiment demonstrates the performance of the minimum $\gamma$-power divergence method in the context of dynamic treatment regimes.

\end{abstract}



\newpage
\section{Introduction}\label{sec1}
Information geometry has attracted much attention from various research interests as a fundamental tool for mathematical science.
Information geometry is originally a differential geometric approach in statistics, which now integrates  machine learning, informatics, statistical physics, data science, artificial intelligence, decision science, and so forth in a field of mathematical science. 
Cf. Rao \cite{rao1945}, Amari \cite{amari1982, amari2016}, Amari \& Nagaoka \cite{amari-nagaoka2007}, Ay, et al. \cite{ay2017} and Eguchi \& Komori \cite{eguchi-komori2022}. 
We would like to focus on further development of information geometry to provide new approaches for statistics and machine learning.  
In particular, we aim at providing the interface between information geometry and machine learning.
In machine learning, research on data learning algorithms is central, in which there are
many concepts inspired by the function of the biological brain and invariants to describe the equilibrium state in a macroscopic system. They have been utilized for elucidating numerical and probabilistic behaviors, and statistical properties of data learning algorithms
and information geometry can provide a unified view of these various studies from a geometric perspective.

This paper introduces an information geometric framework for reinforcement learning.
In particular, we focus on 
a direction of dynamic treatment resumes (DTR), which has been broadly requested to give support for the medical decision, cf.
Chakraborty  \& Moodie \cite{chakraborty-moodie2013} and Chapter 8 in \cite{eguchi-komori2022}.
For  a personalized treatment choice,  it is necessary to get a better decision based on huge digital information  such as medical imaging and mobile health instruments for individuals.
Thus, the framework of DTR is mostly applied from the general framework of   reinforcement learning.
However the Markovian property is not valid in DTR, and the conditioning on all the history variables is necessary.
For this, a  sequential multiple assignment randomized trial (SMART) design is formulated and assumed  for individual trajectories recorded from initial observations to final outcomes, cf. Murphy \cite{murphy2005} for a general framework, and
Qian et al. \cite{qian-etal2021} for the excursion effect based on a long-term time dependency of the treatment effect.
A study of dynamic treatment resumes is in the direction of evidence-based decision science, which is parallel to
political sciences and economical studies. 
 Dynamic treatment resumes for personalized medicine are being actively attempted to increase the expected outcome for individuals
with  chronic diseases over  long-term care, taking into account their personal history, physiological characteristics, disease status, etc. This trend is expected to expand to various fields, including evidence-based policy making (EBMP), cf. Imai  and Yamamoto \cite{imai-yamamoto2013}.
The key to its development lies in building a strong relationship between reinforcement learning and causal inference.


In DTR a Q-function is defined by the conditional expectation of an outcome given by  covariate and treatment. 
We introduce information divergence in the space of all the Q-functions, and the minimum divergence estimation is discussed from a point of
statistical property.
In particular, we investigate the statistical behavior of the $\beta$-power and $\gamma$-power divergence in the Q-function space focusing on a semiparametric multiplicative model.
For this, the policy equivalence relation is formulated, in which we observe that the $\gamma$-power divergence is well-defined
on  the quotient set.
As a main result, this property shows that the minimum $\gamma$-power divergence estimator is semiparametric consistent.

The paper is organized as follows.
Section \ref{sec2} gives a basic framework for DTR with value functions and Q-functions.  
The policy equivalence is defined by the motivation associated with the main objective of DTR.
Section \ref{sec3} introduces the  $\gamma$-power divergence in the Q-function space.
In particular, a special property of the  $\gamma$-power divergence for the policy equivalence is discussed.
Further, specific choices of the power index $\gamma$ are associated with geometric and harmonic means for a Q-function.
In Section \ref{sec4}, the $\beta$-power divergence is introduced as a typical example of $U$-divergence.
Section \ref{sec5} discusses the estimation for the optimal policy based on observed trajectories.
As a main result, the minimum $\gamma$-power and $\beta$-power divergence estimators are discussed in the context of semiparametric consistency.
Section \ref{sec6} compares the Q-learning and the proposed method, also the relationship for the maximum likelihood method is pointed out.
In section \ref{sec7}, we have a simulation study. 
The main result in Section 5 is confirmed and compared in a simple probability setting with the means and root mean squares estimates.
Section \ref{sec8} gives discussion for the present conclusion and future problems.

 




\section{Framework}\label{sec2}


For convenience, we discuss a situation of single stage, in which a situation of multi-stage will be discussed later. 	
Let $\bm X$, $A$, and $Y$ be a state variable, an action variable, and a reward variable
with values in $\cal X$, $\cal A$, and $\cal Y$, respectively.  
Assume that the action space $A$ is finite discrete 
and that $Y$ has a nonnegative value with probability one.
Then, the joint density-mass function of $(X,A,Y)$ is written by
\begin{eqnarray}\label{under}
p(\bm x,a,y)=p(y|\bm x,a)p(a|\bm x)p(\bm x),
\end{eqnarray}
where $p(y|\bm x,a)$ is the conditional density function of $Y$ given $\bm X=\bm x$ and $A=a$, $p(a|\bm x)$ is the conditional probability of $A$ given $\bm X=\bm x$. 
In a context of dynamic treatment regimes, $\bm X$, $A$, and $Y$ are said to be  covariate, treatment and outcome variables, respectively. 
In the setting, a deterministic policy function $\cal D$ is defined as a mapping of $\cal X$ into ${\cal A}$. 
The main objective in the dynamic treatment regimes is to find a good policy ${\cal D}(\bm x)$ in the sense that a subject with the covariate $\bm x$
should be received an effective treatment ${\cal D}(\bm x)$ in $\cal A$.
The value function for a policy function ${\cal D}(x)$ is defined by the expected outcome taken by the policy function as a measure of optimality, that is  ${\mathbb V}({\cal D}) = {\mathbb E}_{\cal D}[Y]$, where $\mathbb E_{\mathbb D}$ denotes the expectation with respect to the probability distribution of ${\cal D}(\bm X)$.  
Under a set of regularity conditions, the value function is written as
\begin{eqnarray}\label{value}
        {\mathbb V}({\cal D})  
   = \mathbb E\Big[\frac{{\mathbb I}\big(A={\cal D}(\bm X)\big)}{p(A|\bm X)}Y\Big],
\end{eqnarray}
where $\mathbb E$ denotes the expectation under the underlying distribution with $p(\bm x,a,y)$ in \eqref{under} and $\mathbb I$ denotes the indicator function. 
The optimal policy  is defined by  
\begin{eqnarray}\label{opt0}
{\cal D}^{\rm opt}=\argmax_{{\cal D}
}\mathbb V({\cal D}), 
\end{eqnarray}
cf. Bellman \cite{bellman1957}.
The Q-function is defined by the conditional expectation of $Y$ given $(\bm  X, A)$ as
\begin{eqnarray}\nonumber 
Q(\bm x,a)=\mathbb E[Y|\bm X=\bm x,A=a],
\end{eqnarray}
cf. Watkins \cite{watkins1989} and Sutton and Barto \cite{sutton-arto2018}.
\begin{proposition}\label{prop1}
The optimal policy ${\cal D}^{\rm opt}$ defined in \eqref{opt0} satisfies
\begin{eqnarray}\label{opt}
  {\cal D}^{\rm opt}(\bm x)= \argmax_{a\in{\cal A}} Q(\bm x,a)
\end{eqnarray}
for almost everywhere $\bm x$.
\end{proposition}
\begin{proof}
It is noted that 
\begin{eqnarray}\nonumber
{\mathbb V}({\cal D})&=&\int_{\cal X}\sum_{a\in{\cal A}}\int_{\cal Y}
\frac{{\mathbb I} (a={\cal D}(\bm x) )}{p(a|\bm x)}yp(y|\bm x,a)p(a|\bm x)p(\bm x)dyd\bm x
\nonumber\\[2mm]
&=&\int_{\cal X} Q(\bm x,{\cal D}(\bm x))p(\bm x) d\bm x.
\end{eqnarray}
By definition,
\begin{eqnarray}\nonumber
{\mathbb V}({\cal D}^{\rm opt})-{\mathbb V}({\cal D})
 &=& \int_{\cal X} \{Q(\bm x,{\cal D}^{\rm opt}(\bm x))-Q(\bm x,{\cal D}(\bm x))\}p(\bm x) d\bm x
\nonumber.
\end{eqnarray}
is nonnegative for any policy $\cal D$.
Hence, 
$ 
{\cal D}^{\rm opt}(\bm x)
$ must be the mode of $Q(\bm x, a)$ with respect to $a$ of $\cal A$, 
which concludes \eqref{opt}.
\end{proof}
In accordance, this property \eqref{opt} tells that the optimal policy has a simple construction by the Q-function, from which the Q-learning has exploited with a feasible computation complexity.     

Let $\cal Q$ be the space of all the Q-functions.
We introduce a basic relation in $\cal Q$ as follows.
\begin{definition}\label{def1}
We say that $Q_0(\bm x,a)$ and $Q_1(\bm x,a)$ are policy-equivalent if there exists a function $\eta(\bm x)>0$ defined on $\cal X$ such that
\begin{eqnarray}\nonumber 
Q_1(\bm x,a)=\eta(\bm x) Q_0(\bm x,a),
\end{eqnarray}
and write $Q_1\sim Q_0$. 
\end{definition}
If $Q_1\sim Q_0$, then the  policy functions associated with $Q_0 $ and $Q_1$  are equal.
This is because
\begin{eqnarray}\nonumber
   \argmax_{a\in{\cal A}}Q_1(\bm x,a)=\argmax_{a\in{\cal A}}Q_0(\bm x,a).
\end{eqnarray}
We write $[Q]=\{Q_1\in{\cal Q}: Q_1\sim Q\}$, called the equivalence class. 
Our main objective is to estimate the optimal policy under a model for the true density function \eqref{under}.
Hence, we do not need to estimate the true Q-function $Q_{\rm true}(\bm x,a)$ itself, but the mode   
 of $Q_{\rm true}(\bm x,a)$ in $a$ of $\cal A$.
%

The framework for a multiple stage is straightforward from that for the single stage.
Let $(\bm X_1,A_1,Y_1,...,\bm X_{T},A_T,Y_T)$ be a trajectory for $T$ stage sequence and ${\cal D}=({\cal D}_1,...,{\cal D}_T)$ 
the vector of $T$ policy functions.  
Thus, the $t$-th history vector is written as 
$$\bm H_t=(\bm X_1,A_1,....,\bm X_{t-1},A_{t-1},\bm X_t)$$ 
and
the $t$-th policy function is ${\cal D}_t(\bm H_t)=A_t$.  
Then, the $t$-th Q function is defined by
 \begin{eqnarray}\nonumber
Q_t(\bm h_t,a_t)=\mathbb E\big[Y_t+\max_{a_{t+1}\in{\cal A}_{t+1}}Q_{t+1}(\bm H_{t+1}, a_{t+1})\big|\bm H_t=\bm h_t,A_t=a_t\big].
\end{eqnarray}
In these settings, the policy equivalence relation can be defined and Proposition \ref{prop1} still holds at any stage $t, 1\leq t\leq T$.

\section{$\gamma$-Power divergence}\label{sec3}

Let us introduce a class of information divergence in the Q-function space that is different from
the probability density functions. 
We define the $\gamma$-power divergence defined on ${\cal Q}\times{\cal Q}$ that is
written by
\begin{eqnarray}\label{gamma}
D_\gamma(Q_0,Q_1)=H_\gamma(Q_0,Q_1)-
H_\gamma(Q_0,Q_0),
\end{eqnarray}
where
\begin{eqnarray}\label{entropy}
H_\gamma(Q_0,Q_1)=
-\frac{1}{\gamma}\mathbb E\Bigg[
\frac{ \sum_{a\in{\cal A}} Q_0(\bm X, a)Q_1(\bm X, a)^{\gamma}}
{\big\{ \sum_{a\in{\cal A}}Q_1(\bm X, a)^{1+\gamma}\big\}^{\frac{\gamma}{1+\gamma}}}
\Bigg],
\end{eqnarray}
called the $\gamma$-power  entropy, cf. Fujisawa \& Eguchi \cite{fujisawa-eguchi2008} and Eguchi \& Kato \cite{eguchi-kato2010}. 
It follows from the H\"{o}lder inequality that $D_\gamma(Q_0,Q_1)\geq0$ for all $\gamma\geq0$.
We observe a further property. 
\begin{theorem}\label{thm1}
Let $D_\gamma(Q_0,Q_1)$ be the $\gamma$-power divergence defined by \eqref{gamma}.
Then,  for any $\gamma$ of $\mathbb R$,
\begin{eqnarray}\label{nn}
D_\gamma(Q_0,Q_1)\geq0
\end{eqnarray}
 with equality  if and only if $Q_0\sim Q_1$,
where $\sim$ denotes the policy equivalence relation defined in Definition \ref{def1}. 
\end{theorem}
\begin{proof}
Define a function as $V(R)=-\frac{1}{\gamma}R^{\frac{\gamma}{1+\gamma}}$
 for $R>0$.
Then, $$\frac{d^2}{dR^2}V(R)=\frac{1}{(1+\gamma)^2}R^{-\frac{\gamma+2}{1+\gamma}}, $$and thus
$V(R)$ is convex in $R$ for any $\gamma\in\mathbb R$.
It can be written that  
\begin{eqnarray}\label{eqn1}
D_\gamma(Q_0,Q_1)=
\int_{\cal X}
\sum_{a\in{\cal A}}Q_0(\bm x, a)\{V(R_1(\bm x,a))-V(R_0(\bm x,a))\}p(\bm x)d\bm x,
\end{eqnarray}
where
$
R_j(\bm x,a)={Q_j(\bm x, a)^{\gamma+1}}/
{ \sum_{a^\prime\in{\cal A}}Q_j(\bm x, a^\prime)^{1+\gamma}}
$
for $j=0,1$.
From the convexity of $V(R)$,
\begin{eqnarray}\label{ineq11}
D_\gamma(Q_0,Q_1)\geq
\int_{\cal X}
\sum_{a\in{\cal A}}Q_0(\bm x, a)\frac{dV(R_0(\bm x,a))}{dR}\{R_1(\bm x,a)-R_0(\bm x,a)\}p(\bm x)d\bm x,
\end{eqnarray}
of which the right side is given by
\begin{eqnarray}\nonumber
\int_{\cal X}\Big\{\sum_{a\in{\cal A}}Q_0(\bm x, a)^{1+\gamma}\Big\}^{\frac{1}{1+\gamma}}
\sum_{a\in{\cal A}}\{R_1(\bm x,a)-R_0(\bm x,a)\}p(\bm x) d\bm x,
\end{eqnarray}
which vanishes since $\sum_{a\in{\cal A}}\{R_1(\bm x,a)-R_0(\bm x,a)\}=0$.
This concludes \eqref{nn}. 

Next, assume that   $Q_0\sim Q_1$. 
Then, we observe that $R_1(\bm x,a)=R_0(\bm x,a)$, which implies $D_\gamma(Q_0,Q_1)=0$ due to \eqref{eqn1}.
Inversely, if $D_\gamma(Q_0,Q_1)=0$, then it must satisfy $R_1(\bm x,a)=R_0(\bm x,a)$  in \eqref{ineq11} from the convexity of $V(R)$.   
This implies $Q_1(\bm x, a)=\eta^*(\bm x)Q_0(\bm x, a)$, where
\begin{eqnarray}\nonumber
\eta^*(\bm x)=\bigg\{\frac{\sum_{a\in{\cal A}}Q_1(\bm x, a)^{1+\gamma}}
{ \sum_{a\in{\cal A}}Q_0(\bm x, a)^{1+\gamma}}\bigg\}^{\frac{1}{1+\gamma}}.
\end{eqnarray}
That is, $Q_0\sim Q_1$, which completes the proof.
\end{proof}
It is noted that the $\gamma$-power divergence in the probability density function space is not well thought for the negative power index, however
$D_\gamma(Q_0,Q_1)$ is well defined, and the specific choices of $\gamma=-1,-2$ will be discussed later.
Further, the invariance under $\sim$ is explored in the following. 
\begin{corollary}\label{cor1}
If $Q_2\sim Q_1$, then 
\begin{eqnarray}\label{result}
D_\gamma(Q_0,Q_2)=D_\gamma(Q_0,Q_1).
\end{eqnarray}
\end{corollary}
\begin{proof}
From the assumption, there exists a function $\eta(\bm x)$ such that $Q_2(\bm x,a)=\eta(\bm x)Q_1(\bm x,a)$.
Hence, we observe that
\begin{eqnarray}\label{eq1}
D_\gamma(Q_0,Q_2)=
-\frac{1}{\gamma}\int_{\cal X}
\frac{ \sum_{a\in{\cal A}} Q_0(\bm x, a)\{\eta(\bm x)Q_1(\bm x, a)\}^{\gamma}}
{\{ \sum_{a\in{\cal A}}\{\eta(\bm x)Q_1(\bm x, a)\}^{1+\gamma}\}^{\frac{\gamma}{1+\gamma}}}
p(\bm x)d\bm x-H_\gamma(Q_0,Q_0).
\end{eqnarray}
It is noted that two $\eta(\bm x)$'s in the first term of the right side of \eqref{eq1} are canceled in the integral sign. 
This concludes \eqref{result}.

\end{proof}
We remark from Theorem \ref{thm1} that $D_\gamma(Q_0,Q_1)$ does not satisfy the distinguishability over $\cal Q$, but
satisfies the distinguishability over ${\cal Q}/\sim$.
In this sense, the $\gamma$-power divergence can be more efficiently employed for estimating the optimal policy function than $U$-divergence, which will be given in a subsequent discussion.  
It is noted that if $\gamma$ is taken as a limit to $0$, then 
\begin{eqnarray}\label{nKL}
\hspace{-5mm}\lim_{\gamma\rightarrow0}D_\gamma(Q_0,Q_1) =
 \sum_{a\in{\cal A}}\mathbb E\bigg[Q_0(\bm X,a)\Big\{\log \frac{Q_0(\bm X,a)}{\sum_{a'\in{\cal A}}Q_0(\bm X,a')}-\log \frac{Q_1(\bm X,a)}{\sum_{a'\in{\cal A}}Q_1(\bm X,a')}\Big\}\bigg],
\end{eqnarray}
say $D_{\rm nKL}$, which is the normalized KL divergence that is different from  the extended KL divergence. 
We note that equation \eqref{nKL} is derived to be the limit of \eqref{eqn1} as $\gamma$ goes to $0$.


Furthermore, the $\gamma$-power divergence equips with a remarkable relation to the value function in the limiting for $\gamma$ to $\infty$ as follows.
\begin{theorem}\label{thm2}
Let $\mathbb V_0$ be a value function generated by a Q-function $Q_0$ such that
\begin{eqnarray}\label{V2Q}
 \mathbb V_0({\cal D})=\mathbb E[ Q_0(\bm X, {\cal D}(\bm X))].
\end{eqnarray} 
Then, 
\begin{eqnarray}\label{Value}
\lim_{\gamma\rightarrow\infty}\gamma\>D_\gamma(Q_0,Q_1)= \mathbb V_0({\cal D}_0) - \mathbb V_0({\cal D}_1)
\end{eqnarray}
where ${\cal D}_j(\bm x)=\argmax_{a\in{\cal A}}Q_j(\bm x,a)$ for $j=0,1$.
\end{theorem}
\begin{proof}
We observe that the $\gamma$-power cross entropy satisfies
\begin{eqnarray}\nonumber
\lim_{\gamma\rightarrow\infty}\gamma\> H_\gamma (Q_0,Q_1)=
-\lim_{\gamma\rightarrow\infty} \int_{\cal X}
\frac{ \sum_{a\in{\cal A}} Q_0(\bm x, a)\Big\{\displaystyle\frac{Q_1(\bm x, a)}{Q_1(\bm x, {\cal D}_1(\bm x))}\Big\}^{\gamma}}
{\Big[ \sum_{a\in{\cal A}}\Big\{\displaystyle\frac{Q_1(\bm x, a)}{Q_1(\bm x, {\cal D}_1(\bm x))}\Big\}^{1+\gamma}\Big]^{\frac{\gamma}{1+\gamma}}}
\>p(\bm x)d\bm x
\end{eqnarray}
which is equal to
$
- \int_{\cal X}
 \sum_{a\in{\cal A}} Q_0(\bm x, a){\mathbb I}(a= {\cal D}_1(\bm x))p(\bm x)d\bm x
.$ 
This is written as
\begin{eqnarray}\nonumber
- \int_{\cal X}Q_0(\bm x,{\cal D}_1(\bm x))p(\bm x)d\bm x
\end{eqnarray}
which is nothing but $-\mathbb V_0({\cal D}_1)$ due to \eqref{V2Q}.
Similarly, we observe
\begin{eqnarray}\nonumber
\lim_{\gamma\rightarrow\infty}\gamma\> H_\gamma(Q_0,Q_0)=-\mathbb V_0({\cal D}_0),
\end{eqnarray}
 which concludes \eqref{Value}.

\end{proof}


The equation \eqref{Value} is an ultimate form that depends only on two policies ${\cal D}_0$ and ${\cal D}_1$ rather than a form depending on $Q_0$ and $Q_1$. 
Theorem \ref{thm2} directly supports  that  ${\cal D}_1$ is equal to 
the optimal policy ${\cal D}^{\rm opt}$ in \eqref{opt} with respect to the value function ${\mathbb V}_0$ since $${\mathbb V}_0({\cal D}_1)\geq {\mathbb V}_0({\cal D}_2)$$ for any ${\cal D}_2$.
Consequently, the minimum $\gamma$-power divergence is equivalent to the maximum value function in this limiting sense.

Next, we have a look at a close relationship of the geometric mean of Q-functions taking another limit of $\gamma$.
  
\begin{theorem}\label{thm3}
\begin{eqnarray}\nonumber
 \lim_{\gamma\rightarrow-1}\ m^{\frac{1}{1+\gamma}}D_\gamma(Q_0,Q_1) \label{GM}
=
\mathbb E\bigg[\frac{1}{m}\sum_{a\in{\cal A}}\frac{Q_0(\bm X,a)}{Q_1(\bm X,a)}{\rm GM}_{Q_1}(\bm X)
-{\rm GM}_{Q_0}(\bm X)\bigg],
\end{eqnarray}
where $m$ is the cardinal number of $\cal A$ and 
$ 
{\rm GM}_{Q}(\bm x)=\prod_{a\in{\cal A}}Q(\bm x,a)^{\frac{1}{m}}.
$

\end{theorem}
\begin{proof}
We observe 
\begin{eqnarray}\nonumber
            \lim_{\gamma\rightarrow -1}m^{\frac{1}{1+\gamma}}H_\gamma(Q_0,Q_0)
&=&\lim_{\gamma\rightarrow -1}
\int_{\cal X}{\exp \Big\{{\frac{1}{\gamma+1}}\log\frac{1}{m}\sum_{a\in{\cal A}} {Q_0(\bm x,a)}^{\gamma+1}\Big\}}p(\bm x)d\bm x\nonumber\\[3mm]
&=&\lim_{\gamma\rightarrow -1}
\int_{\cal X}{\exp \Big\{\frac{\frac{1}{m}\sum_{a\in{\cal A}} {Q_0(\bm x,a)}^{\gamma+1}\log Q_0(\bm x,a)}{\frac{1}{m}\sum_{a\in{\cal A}} {Q_0(\bm x,a)}^{\gamma+1}}\Big\}}p(\bm x)d\bm x
\nonumber\\[3mm]
&=& 
\int_{\cal X}{\exp \Big\{ {\frac{1}{m}\sum_{a\in{\cal A}}\log Q_0(\bm x,a) \Big\}}}p(\bm x)d\bm x
\nonumber\\[.6mm]\nonumber 
&=& 
 \int_{\cal X}\Big(\prod_{a\in{\cal A}}  Q_0(\bm x,a)\Big)^{\frac{1}{m}}p(\bm x)d\bm x.
\end{eqnarray}
Similarly,
\begin{eqnarray}\nonumber
            \lim_{\gamma\rightarrow -1}m^{\frac{1}{1+\gamma}}H_\gamma (Q_0,Q_1)
= 
  \int_{\cal X}\sum_{a\in{\cal A}}\frac{Q_0(\bm x,a)}{Q_1(\bm x,a)}\Big(\prod_{a\in{\cal A}}  Q_0(\bm x,a)\Big)^{\frac{1}{m}}p(\bm x)d\bm x.
\end{eqnarray}
Therefore, we conclude \eqref{GM}.
\end{proof}

We remark that the well-known inequality between the arithmetic and geometric means implies 
\begin{eqnarray}\nonumber
\mathbb E\bigg[\frac{1}{m}\sum_{a\in{\cal A}}\frac{Q_0(\bm X,a)}{Q_1(\bm X,a)}\bigg] \geq\mathbb E\bigg[ \frac{{\rm GM}_{Q_0}(\bm X)}{{\rm GM}_{Q_1}(\bm X)}\bigg],
\end{eqnarray}
which is closely related to the limit in Theorem \eqref{thm3}.

Next, let us consider an  interesting relation of the $\gamma$-power entropy with the harmonic mean.
\begin{theorem}\label{thm4}
When $\gamma=-2$, then the diagonal $\gamma$-power entropy is equal to the expected harmonic mean of the Q-function,
that is 
\begin{eqnarray} \label{g2hm}
H_\gamma(Q,Q)
=
\frac{m}{2}\mathbb E\big[{\rm HM}_{Q}(\bm X)\big],
\end{eqnarray}
where
\begin{eqnarray}\nonumber
{\rm HM}_{Q}(\bm x)=\Big[m\sum_{a\in{\cal A}}\frac{1}{Q(\bm x,a)}\Big]^{-1}.
\end{eqnarray}

\end{theorem}
\begin{proof}
By definition, when $\gamma=-2$, 
\begin{eqnarray}\nonumber
H_\gamma(Q_0,Q_1)=
\frac{1}{2}\int_{\cal X}
\frac{ \sum_{a\in{\cal A}} Q_0(\bm x, a)Q_1(\bm x, a)^{-2}}
{\{ \sum_{a\in{\cal A}}Q_1(\bm x, a)^{-1}\}^{2}}
p(\bm x)d\bm x
\end{eqnarray}
from \eqref{entropy}.  This concludes \eqref{g2hm}.
\end{proof}
Consider an  inequality associated with the arithmetic and harmonic means as 
\begin{eqnarray}\label{prototype}
   \sum_{a\in{\cal A}}\frac{Q_0(\bm x,a)}{Q_1(\bm x,a)}w(Q_1(\bm x,a))\geq 
\Big[\sum_{a\in{\cal A}}\frac{Q_1(\bm x,a)}{Q_0(\bm x,a)}w(Q_1(\bm x,a))\Big]^{-1}
\end{eqnarray}
for any weight function $w(Q_1(\bm x,a))$ satisfying $\sum_{a\in{\cal A}} w(Q_1(\bm x,a))=1$.  If  we fix  as
\begin{eqnarray}\label{weight1}
  w(Q_1(\bm x,a))=\frac{1}{Q_1(\bm x,a)}\Big(\sum_{a'\in{\cal A}} \frac{1}{Q_1(\bm x,a')}\Big)^{-1},
\end{eqnarray}
then  the inequality \eqref{prototype} reduces to
\begin{eqnarray}\nonumber
 \Big(  \sum_{a\in{\cal A}} \frac{Q_0(\bm x,a)}{Q_1(\bm x,a)^2}\Big)\Big(\sum_{a\in{\cal A}}  \frac{1}{Q_1(\bm x,a)}\Big)^{-1}
\geq \Big(\sum_{a\in{\cal A}}  \frac{1}{Q_0(\bm x,a)}\Big)^{-1}\Big(\sum_{a\in{\cal A}}  \frac{1}{Q_1(\bm x,a)}\Big),
\end{eqnarray}
or equivalently,
\begin{eqnarray}\label{ineq}
 \Big( \sum_{a\in{\cal A}} \frac{Q_0(\bm x,a)}{Q_1(\bm x,a)^2}\Big)\Big(\sum_{a\in{\cal A}}  \frac{1}{Q_1(\bm x,a)}\Big)^{-2}
\geq \Big(\sum_{a\in{\cal A}} \frac{1}{Q_0(\bm x,a)}\Big)^{-1}.
\end{eqnarray}
Note that the weight function $w(Q_1(\bm x,a))$ in \eqref{weight1} is essential, which is viewed as a harmonic weight.
In summary, when $\gamma=-2$, then the $\gamma$-power divergence 
is equivalent to the expectation for both sides in the inequality \eqref{ineq}, that is
\begin{eqnarray}\nonumber 
 D_\gamma(Q_0,Q_1)=\frac{m}{2} \mathbb E\Big[ \sum_{a\in{\cal A}} \frac{Q_0(\bm X,a)}{Q_1(\bm X,a)^2}{\rm HM}_{Q_1}(\bm X)^{2}\Big]
-\frac{m}{2} \mathbb E\Big[ {\rm HM}_{Q_0}(\bm X)\Big].
\end{eqnarray}
Following these, the class of $\gamma$-power divergence is associated with the arithmetic, geometric and harmonic means of Q-functions.

\section{$\beta$-Power divergence}\label{sec4}

We now consider $U$-divergence in the Q-function space $\cal Q$.    
Let $U$ be a strictly increasing and convex function defined on $\mathbb R$.
Then,  an information divergence defined on ${\cal Q}\times{\cal Q}$ 
is introduced as
\begin{eqnarray}\nonumber
D_{U}(Q_0,Q_1)=H_U(Q_0,Q_1)-H_U(Q_0,Q_0),
\end{eqnarray}
called $U$-divergence, where 
\begin{eqnarray}\nonumber
H_{U}(Q_0,Q_1)=\sum_{a\in{\cal A}}\mathbb E \big[ U(u^{-1}(Q_1(\bm X,a))-Q_0(\bm X,a){u^{-1}(Q_1(\bm X,a))}\big]
\end{eqnarray}
with $u=U^\prime$.  
We note from the convexity of $U$ that 
\begin{eqnarray}\nonumber U(u^{-1}(q_1))-U(u^{-1}(q_0))-q_0\{{u^{-1}(q_1)}-{u^{-1}(q_0)}\}\geq0
\end{eqnarray}
for any nonnegative numbers $q_0$ and $q_1$ and the equality holds if and only if $q_0=q_1$. This implies that $D_U(Q_0,Q_1)$
is well-defined as an information divergence in $\cal Q$, that is, $D_U(Q_0,Q_1)\geq0$ with equality if and only if $Q_0=Q_1$.  Let
\begin{eqnarray}\nonumber
U(t)=\frac{(1+\beta t)^{\frac{\beta+1}{\beta}}}{\beta+1}
\end{eqnarray}
with a power index $\beta\in\mathbb R$.  Then, as a typical example of $U$-divergence,  the $\beta$-power divergence is introduced as
\begin{eqnarray}\nonumber
D_{\beta}(Q_0,Q_1)=H_{\beta}(Q_0,Q_1)-H_{\beta}(Q_0,Q_0),
\end{eqnarray}
where
\begin{eqnarray}\nonumber
H_{\beta}(Q_0,Q_1)=\sum_{a\in{\cal A}}\mathbb E\Big[ \frac{Q_1(\bm X,a)^{\beta+1}}{\beta+1}
-\frac{Q_0(\bm X,a)Q_1(\bm X,a)^\beta}{\beta}\Big].
\end{eqnarray}
Unfortunately, the $U$-divergence does not satisfy the property of the policy equivalence that is satisfied by the $\gamma$-power divergence.
Assume $Q_0\sim Q_1$, that is, there exists $\eta(\bm x)$ such that  $Q_1(\bm x,a)=\eta(\bm x) Q_0(\bm x,a)$.
Then,  for example, we observe in the $\beta$-power divergence that 
\begin{eqnarray}\nonumber
D_\beta(Q_0,Q_1)=\sum_{a\in{\cal A}}\mathbb E\Big[Q_0(\bm X,a)^{\beta+1} 
\Big\{ \frac{\eta(\bm X)^{\beta+1}-1}{\beta+1}
-\frac{\eta(\bm X)^\beta-1}{\beta}\Big\}\Big]
\end{eqnarray}
 which does not always vanish.

\section{Estimation for the optimal policy}\label{sec5}

\subsection{Minimum $\gamma$-power divergence}

We consider the estimation for optimal policy employing the $\gamma$-power divergence.
For this, we assume a semiparametric model for the Q-function as
\begin{eqnarray}\label{Q-model}
{\cal M}=\{ Q(\bm x, a, \bm\psi):=\exp\{f(\bm x)+g(\bm x,a,\bm\psi)\}: f\in{\cal F},\psi\in{\Psi}\},
\end{eqnarray} 
where $f(\bm x)$ is a nonparametric component; $g(\bm x,a,\bm\psi)$ is a parametric component with
a parameter vector $\bm\psi$.
Thus, the component $f(\bm x)$ models the part of the conditional expectation of the outcome $Y$ given only the covariate $\bm x$ that is independent of the conditioning for the treatment $A=a$.
On the other hand, the component  $g(\bm x,a,\bm\psi)$ models the important part depending on both
conditionings of the covariate $\bm X=\bm x$ and  the treatment $A=a$.
In effect, the optimal policy is derived only by the parametric component  as
\begin{eqnarray}
{\cal D}^{\rm opt}(\bm x,\bm\psi) =\underset{a\in{\cal A}}{\rm argmax}\ g(\bm x,a,\bm\psi).
\end{eqnarray} 
In this sense, the nonparametric component $f(\bm x)$ is viewed as a nuisance function.

For a dataset $\{(\bm X_i,A_i,Y_i) :i=1,...,n\}$ from the model $\cal M$ 
 the $\gamma$-power  loss function for $\bm\psi$ is defined by
\begin{eqnarray}\label{gamma-loss}
L_\gamma({\bm\psi})=- \frac{1}{\gamma}\frac{1}{n}\sum_{i=1}^n
\frac{Y_i}{p(A_i|\bm X_i)}
\displaystyle \frac{\exp\{\gamma g(\bm X_i,A_i,{\bm\psi})\}}{
\{\sum_{a \in{\cal A}}\exp\{(\gamma+1) g(\bm X_i,a,{\bm\psi})\}\}^{\frac{\gamma}{1+\gamma}}},
\end{eqnarray}
where the conditional treatment probability $p(a|\bm x)$ is assumed to be a known probability function of $A$ given $\bm X=\bm x$.
In a case that $p(a|\bm x)$ is unknown, it is usually estimated by modeling in a parametric manner, for example, a logistic model based on $\{(\bm X_i,A_i):i=1,...,n\}$.
Cf. Henmi \& Eguchi \cite{henmi2004} for a paradoxical aspect, and Wallace \& Moody \cite{wallace2015} for the
inverse probability weighing and the balancing condition.
We note that the loss function \eqref{gamma-loss} is viewed as the outcome weighted form,
which is nicely discussed with the suport vector machine in Zhao et al. \cite{zhao2015}.

The proposed estimator is given by
\begin{eqnarray}\label{psi-est}
\hat{\bm\psi}_\gamma= \argmin_{\bm\psi\in\Psi}L_\gamma({\bm\psi}),
\end{eqnarray}
called the minimum $\gamma$-power divergence estimator ($\gamma$-MDE).
The estimated policy function for a given covariate $\bm x$ is provided  by 
$$\hat{\cal D}_\gamma(\bm x)=\argmax_{a\in{\cal A}}g(\bm x,a,{\hat{\bm\psi}_\gamma})$$
via plugged-in the $\gamma$-MDE $\hat{\bm\psi}_\gamma$ to $\bm\psi$.

\begin{theorem}\label{thm5}
Let $\{(\bm X_i,A_i,Y_i) :i=1,...,n\}$ be a random sequence with the Q-function of $\cal M$ with unknown parameter 
$\bm\psi_0$ and nuisance function $f(\bm x)$. 
Then, the $\gamma$-MDE $\hat{\bm\psi}_\gamma$ defined in  \eqref{psi-est}
is consistent with $\bm\psi_0$ irrelevant to  $f(\bm x)$. 

\end{theorem}
\begin{proof}
We observe that 
\begin{eqnarray}\nonumber 
\mathbb E[L_\gamma({\bm\psi})]=-  \frac{1}{\gamma}\mathbb E\Big[
\frac{Y}{p(A|\bm X)}
\displaystyle \frac{\exp\{\gamma g(\bm X ,A ,{\bm\psi})\}}{
\{\sum_{a \in{\cal A}}\exp\{(\gamma+1) g(\bm X ,a,{\bm\psi})\}\}^{\frac{\gamma}{1+\gamma}}}\Big]
\end{eqnarray}
which is written as
\begin{eqnarray}\nonumber 
-  \frac{1}{\gamma}\mathbb E\Big[
{Q(\bm X,A,\bm\psi_0)}
\displaystyle \frac{Q(\bm X,A,\bm\psi)^\gamma}{
\{\sum_{a \in{\cal A}}Q(\bm X,A,\bm\psi)^{\gamma+1}\}^{\frac{\gamma}{1+\gamma}}}\Big]
\end{eqnarray}
by an argument similar to the proof of Corollary \ref{cor1} since $\exp\{ g(\bm X ,A ,{\bm\psi})\}$ and $Q(\bm x, a,\bm\psi)$ is policy-equivalent.
Consequently,
\begin{eqnarray}\nonumber 
\mathbb E[L_\gamma({\bm\psi})]=H_\gamma(Q(\bm X,A,\bm\psi_0),Q(\bm X,A,\bm\psi))
\end{eqnarray}
which implies 
\begin{eqnarray}\nonumber 
\mathbb E[L_\gamma({\bm\psi})]-\mathbb E[L_\gamma({\bm\psi}_0)]=D_\gamma(Q(\bm X,A,\bm\psi_0),Q(\bm X,A,\bm\psi)).
\end{eqnarray}
This concludes $\bm\psi_0=\argmin_{\bms\psi\in\Psi} \mathbb E[L_\gamma(\psi)]$.
Since $\{(\bm X_i,A_i,Y_i) :i=1,...,n\}$ be a random sequence, $L_\gamma(\psi)$ almost surely converges to $ \mathbb E[L_\gamma(\psi)]$.
Hence, from the continuous mapping theorem, $\hat{\bm\psi}_\gamma$ is almost surely converges to  $\bm\psi_0$.
\end{proof}
Here we observe a remarkable property of the $\gamma$-MDE, in which
the consistency for the parameter $\bm\psi$ holds even if there exists the unknown nuisance function $f(\bm x)$.
In other words,  the $\gamma$-MDE satisfies the semiparametric consistency in the multiplicative semiparametric model \eqref{Q-model}.
This property is similar to the property of  G-estimator, in which  the blip function or the difference of Q-functions,
 cf. Robins \cite{robins2004} and Robins et al. \cite{robins-etal2000}.
Consider a semiparametric additive model as
\begin{eqnarray}\nonumber 
\tilde{\cal M}=\{ Q(\bm x, a, \bm\psi):= f(\bm x)+g(\bm x,a,\bm\psi): f\in{\cal F},\psi\in{\Psi}\}.
\end{eqnarray} 
Then, the G-estimator efficiently eliminates the nuisance function as in the blip function
\begin{eqnarray}\nonumber 
     Q(\bm x, a, \bm\psi)- Q(\bm x, a_0, \bm\psi)=g(\bm x,a,\bm\psi)-g(\bm x,a_0,\bm\psi)
\end{eqnarray} 
with a fixed treatment $a_0$.

The estimating function of the $\gamma$-MDE, or the gradient vector of the summand 
in \eqref{gamma-loss} is given by
\begin{eqnarray}\nonumber 
 {\cal E}_\gamma({\bm\psi}, \bm X,A,Y) =\frac{Y }{p(A|\bm X)}
 \frac{e^{\gamma g(\bm X,A,\bms\psi)}}{R(\bm X,\bm \psi)}
 {\sum_{a \in{\cal A}}w(\bm X,a,\bm\psi)\{\bm G(\bm X ,A  ,{\bm\psi})-\bm G(\bm X ,a ,{\bm\psi})}\}, 
\end{eqnarray}
where $w(\bm x,a,\bm\psi)=\exp\{(\gamma+1) g(\bm x ,a,{\bm\psi})\}$, $\bm G(\bm x,a,\bm\psi)=(\partial/\partial\bm\psi)g(\bm x,a, \bm\psi)$ and
\begin{eqnarray}\nonumber 
 R(\bm X,\bm\psi)=
 \bigg({\sum_{a \in{\cal A}}w(\bm X,a,\bm\psi)}\bigg)^{\frac{\gamma}{\gamma+1}+1}.
\end{eqnarray}
Thus, we observe
\begin{eqnarray}\nonumber 
&&\mathbb E[{\cal E}_\gamma({\bm\psi}, \bm X,A,Y)|\bm X=\bm x]\nonumber\\[2mm]
&=&\sum_{a \in{\cal A}}\sum_{a' \in{\cal A}} \mathbb E\big[
w(\bm X,a',\bm\psi)w(\bm X,a,\bm\psi)\{ \bm G(\bm X ,a'  ,{\bm\psi}) - \bm G(\bm X ,a ,{\bm\psi})  \}\big],
\end{eqnarray}
which vanishes. 
This is because we observe a formula such that 
\begin{eqnarray}\nonumber 
\mathbb E\bigg[\frac{Y }{p(A|\bm X)}F(\bm X,A)\bigg{|}\bm X=\bm x\bigg]=\exp\{f(\bm x)\}
 \sum_{a\in{\cal A}}\exp\{g(\bm x, a,\bm\psi)\}F(\bm x, a). 
\end{eqnarray}
 Hence the estimating function is unbiased.
The $\gamma$-MDE $\hat{\bm\psi}_\gamma$ is given by
the solution of the estimating equation
\begin{eqnarray}\nonumber 
\sum_{i=1}^n{\cal E}_\gamma({\bm\psi}, \bm X_i,A_i,Y_i)={\bf0}.
\end{eqnarray}
The asymptotic variance is given by $\bm J(\bm\psi)^{-1} \bm I(\bm\psi) \{\bm J(\bm\psi)^\top\}^{-1}$
due to the sandwich formula, where
\begin{eqnarray}\nonumber 
\bm I(\bm\psi)= \mathbb E\big[{\cal E}_\gamma({\bm\psi}, \bm X,A,Y){\cal E}_\gamma({\bm\psi}, \bm X,A,Y)^\top\big] 
\end{eqnarray}
and
\begin{eqnarray}\nonumber 
\bm J(\bm\psi)=\mathbb E\Big[\frac{\partial {\cal E}_\gamma({\bm\psi}, \bm X,A,Y)}{\partial \bm\psi^\top}\Big]\end{eqnarray}
There is an interesting problem: which $\gamma$ gives the locally efficient estimator in
the class of all the $\gamma$-MDE? 
However, the problem is unsolved, and so it should be clarified a geometric understanding  for  the nonparametric nuisance space
\begin{eqnarray}\nonumber 
{\cal N} = \{\exp\{f(\bm x)\}:f\in{\cal F}\}
\end{eqnarray} 
in the Q-function space $\cal Q$.

We now consider a  linear model as
\begin{eqnarray}\nonumber 
g(\bm x,a,\bm\psi)=\bm\psi_0^\top \bm s_0(a)+\bm s_1(a)^\top \bm\Psi_1 \bm t(\bm x)
\end{eqnarray} 
where $\bm\psi=(\bm\psi_0,\bm\Psi_1)$.
The first term denotes a main effect for the treatment $A$; the second term denotes an interaction effect between the covariate $\bm X$ and the treatment $A$. 
Thus, in the estimating function ${\cal E}_\gamma(\bm\psi,\bm X,A,Y)$,
\begin{eqnarray}\nonumber 
\bm G(\bm X ,A  ,{\bm\psi})-\bm G(\bm X ,a ,{\bm\psi}) =
\begin{bmatrix}
\bm s_0(A)-\bm s_0(a)\\
(\bm s_1(A)-\bm s_1(a)) \bm t(\bm X)^\top
\end{bmatrix}
\end{eqnarray}
We have a specific choice for the power index $\gamma$ as $-1$, in which the $\gamma$-power divergence is associated with the geometric mean of Q-functions, discussed as in Theorem \ref{thm3}.
Then, the estimating function becomes an unweighted form as
\begin{eqnarray}\nonumber 
 {\cal E}_\gamma({\bm\psi}, \bm X,A,Y)= \frac{Y }{p(A|\bm X)}
 {\exp\{- g(\bm X ,A  ,{\bm\psi})\}}
 \begin{bmatrix}
\bm s_0(A)-\bar{\bm s}_0\\
(\bm s_1(A)-\bar{\bm s}_1) \bm t(\bm X)^\top
\end{bmatrix}
\end{eqnarray}
because the weight function $w(\bm x, a, \bm\psi)$ becomes $1$, where
$\bar{\bm s}_j=m^{-1}\sum_{a\in{\cal A}}\bm s_j(a)$ for $j=0,1$.

We briefly discuss the situation in  a multiple stage.
	Let $\{(\bm X_{1i},A_{1i},Y_{1i},...,\bm X_{Ti},A_{Ti},Y_{Ti}):i=1,...,n\}$ be a random sequence of $n$ trajectories for $T$ stage sequence and ${\cal D}=({\cal D}_1,...,{\cal D}_T)$ 
the vector of $T$ policy functions.  
Thus, the $t$-th history vector is written as $$\bm H_t=(\bm X_1,A_1,....,\bm X_{t-1},A_{t-1},\bm X_t)$$ and the $t$-th policy function is ${\cal D}_t(\bm H_t)=A_t$.  
We assume the SMART condition, cf. Murphy (2005) for detailed discussion.
Then, 
the $t$-th Q function is modeled as
 \begin{eqnarray}\nonumber 
{\cal M}_t = \{Q_t(\bm h_t,a_t):=\exp\{f_t(\bm h_t)+g_t(\bm h_t, a_t, \bm\psi_t)\}:f_t\in {\cal F}_t,\bm\psi_t\in\Psi_t\}.
\end{eqnarray}
At the final stage $T$, the $\gamma$-power  loss function for $\bm\psi_T$ is defined in the same way as in \eqref{gamma-loss}
replacing $(\bm X_i,A_i,Y_i)$'s to $(\bm H_{Ti},A_{Ti},Y_{Ti})$'s.
Assume that the estimated policy $\hat{\cal D}_{t+1}$ is given beforehand the $t$-stage in  a backward manner.
Then the $\gamma$-power  loss function for $\bm\psi_t$ is defined by
\begin{eqnarray}\nonumber 
L^{(t)}_\gamma({\bm\psi})=- \frac{1}{\gamma}\frac{1}{n}\sum_{i=1}^n
\frac{\tilde Y_{ti}}{p_t(A_{ti}|\bm H_{ti})}
\displaystyle \frac{\exp\{\gamma g_t(\bm H_{ti},A_{ti},{\bm\psi}_{t})\}}{
\{\sum_{a_t \in{\cal A}_t}\exp\{(\gamma+1) g_t(\bm H_{ti},a_t,{\bm\psi}_t)\}\}^{\frac{\gamma}{1+\gamma}}},
\end{eqnarray}
where $p_t(a_t|\bm H_t)$ is a known probability function of $A_t$ given $\bm H_t=\bm h_t$ and
$$
\tilde Y_{ti} =Y_{ti}+\hat{Q}_{t+1}(\bm H_{t+1}, \hat{D}_{t+1}(H_{ti}))\big|\bm H_{ti},A_{ti}\big].
$$
The $\gamma$-MDE is given by
\begin{eqnarray}\nonumber 
\hat{\bm\psi}_{t,\gamma}= \argmin_{\bm\psi\in\Psi}L_\gamma^{(t)}({\bm\psi}_t),
\end{eqnarray}
for the backward manner as $t=T,T-1,...,1$.
Finally, we get the estimator of the policy vector $\hat{\cal D}=(\hat{\cal D}_1,...,\hat{\cal D}_T)$ as
$$
\hat{\cal D}_t(\bm h_t)=\argmax_{a_t\in{\cal A}_t }g_t(\bm h_t, a_t, \hat{\bm\psi}_{t,\gamma})
$$ 
for $t=1,...,T$.

\subsection{Minimum $\beta$-power divergence}

Let us employ the $\beta$-power divergence for the estimation of the optimal policy.
We assume a parametric model
\begin{eqnarray}\label{Q-model1}
{\cal M}^*=\{ Q(\bm x, a, \bm \alpha,\bm\psi):=\exp\{f(\bm x,\alpha)+g(\bm x,a,\bm\psi)\}: \alpha \in{\Omega},\psi\in{\Psi}\},
\end{eqnarray} 
where  $\bm\alpha$ and $\bm\psi$ are unknown parameter vectors.
Here it is noted that $\bm\psi$ determines the optimal policy as ${\cal D}^{\rm opt}(\bm x)={\rm argmax}_{a\in{\cal A}} g(\bm x,a,\bm\psi)$; $\bm\alpha$ is irrelevant to any policy determination.
In this sense we call $\bm\psi$ the policy parameter and $\bm\alpha$ the nuisance parameter.
For a dataset $\{(\bm X_i,A_i,Y_i) :i=1,...,n\}$ from the model $\cal M^*$ 
 the $\beta$-power  loss function for $(\bm\alpha,\bm\psi)$ is defined by
\begin{eqnarray}\nonumber 
L_\beta(\bm\alpha,{\bm\psi})&=&\frac{1}{n}\sum_{i=1}^n \frac{{Y_i}}{p(A_i|\bm X_i)}\Big\{-\frac{
\exp\{\beta(f(\bm X_i,\bm\alpha)+g(\bm X_i,A_i,{\bm\psi}))\}}{\beta}\nonumber\\
&&
\hspace{23mm}+\sum_{a\in{\cal A}}\frac{\exp\{(\beta+1)(f(\bm X_i,\bm\alpha)+g(\bm X_i,a,{\bm\psi}))\}}{\beta+1}\Big\},
\end{eqnarray}
where $p(a|\bm x)$ is a known probability function of $A$ given $\bm X=\bm x$.
The proposed estimator is given by
\begin{eqnarray}\nonumber \label{psi-est1}
(\hat{\bm\alpha}_\beta,\hat{\bm\psi}_\beta)= \argmin_{(\bm\alpha,\bm\psi)\in\Omega\times \Psi}L_\beta(\bm\alpha,{\bm\psi}),
\end{eqnarray}
called the minimum $\beta$-power divergence estimator ($\beta$-MDE).
We observe that, if the dataset is from a distribution with $Q(\bm x,a,\bm\alpha_0,\bm\psi_0)$, then
\begin{eqnarray}\nonumber 
\mathbb E[ L_\beta(\bm\alpha,{\bm\psi})]=H_\beta(Q(\bm X, A,\bm\alpha_0,\bm\psi_0),Q(\bm X, A,\bm\alpha,\bm\psi)).
\end{eqnarray}
Hence,  $\mathbb E[ L_\beta(\bm\alpha,{\bm\psi})]-\mathbb E[ L_\beta(\bm\alpha_0,{\bm\psi}_0)]=D_\beta(Q(\bm X, A,\bm\alpha_0,\bm\psi_0),Q(\bm X, A,\bm\alpha,\bm\psi))$.
This concludes the consistency of the $\beta$-MDE $(\hat{\bm\alpha}_\beta,\hat{\bm\psi}_\beta)$ to the true parameter $(\bm\alpha_0,\bm\psi_0)$ by an argument similar to the proof for the $\gamma$-power MHD in Theorem \ref{thm5}.

Consider the $\gamma$-MDE under the model assumption \eqref{Q-model1}.
The $\gamma$-loss function for $L_\gamma(\bm\alpha,\bm\psi)$ is exactly equal to $L_\gamma(\bm\psi)$ in \eqref{gamma-loss} under the semiparametric model $\cal M$ in \eqref{Q-model} since the Q-function 
$Q(\bm x,a,\bm\alpha,\bm\psi)$ is policy-equivalent to $g(\bm x,\bm\psi)$ even when the nuisance function is $f(\bm x,\bm\alpha)$, or $f(\bm x)$.
Hence, the $\gamma$-MHD $\hat{\bm\psi}_\gamma$ for the policy parameter is always consistent with $\bm\psi_0$ even if the nuisance function $f(\bm x,\bm\alpha)$ is misspecified.
On the other hand, the $\beta$-MHD $\hat{\bm\psi}_\beta$ is not consistent with $\bm\psi_0$ in such a
misspecified situation.
Thus, the $\gamma$-MHD and $\beta$-MHD are quite a contrast.
Our main objective is  the estimation for the optimal policy, and hence
the $\gamma$-MHD should be supported from a point of robustness.


\section{Q-learning}\label{sec6}

We overview the Q-learning, which is the regression modeling and analysis with $(\bm x,a)$ as an explanatory variable.
Let ${\cal P}$ be the space of all the conditional density functions of  the outcome $Y$ given the covariate $\bm X=\bm x$ and the treatment $A=a$.

Consider a parametric model in $\cal P$ as
\begin{eqnarray}\label{p-model}
{\cal M}_{\rm C}=\{p(y|\bm x, a,\bm\theta) := \exp\{  y\nu(\bm x,a,\bm\theta)-\kappa(\nu(\bm x,a,\bm\theta))\}: \theta\in\Theta\}
\end{eqnarray}
 with an unknown parameter $\bm\theta$, cf. McCullagh \& Nelder \cite{mccullagh-nelder2019}.
Here the dispersion parameter is assumed to be known.
For the nonnegative outcome $Y$ this model includes Poisson distribution, exponential distribution, inverse Gaussian distribution, and Bernoulli distribution as typical examples. 
For example, a canonical link model is given by
$$\nu(\bm x,a,{\bm\theta})=f(\bm x,\bm\alpha)+g(\bm x, a,\bm\psi),$$
with $ {\bm\theta}=(\bm\alpha,\bm\psi)$, which implies that
\begin{eqnarray}\nonumber 
Q (\bm x,a,\bm\alpha,\bm\psi)=  \kappa^\prime(f(\bm x,\bm\alpha)+g(\bm x, a,\bm\psi))
\end{eqnarray}
due to a basic property of the exponential family \eqref{p-model}.
This modeling is closely related to the model ${\cal M}^*$ defined in \eqref{Q-model1}.
If the link function is a log link, then the model \eqref{p-model} reduces to \eqref{Q-model1}.
Similarly, $\bm\psi$ is the policy parameter; $\bm\alpha$ is the nuisance parameter in the sense that
\begin{eqnarray}\nonumber 
\argmax_{a\in{\cal A}}Q (\bm x,a,\bm\alpha,\bm\psi)=  \argmax_{a\in{\cal A}}g(\bm x, a,\bm\psi)
\end{eqnarray}
due to the monotonicity of $\kappa^\prime(\cdot)$.
For a given dataset $\{\bm X_i,A_i,Y_i):i=1,...,n\}$, the log-likelihood function for the parametric model  is essentially given by
\begin{eqnarray}\nonumber 
L(\bm\theta)= \sum_{i=1}^n \{Y_i \nu (\bm X_i,A_i,\bm\theta)-\kappa(\nu (\bm X_i,A_i,\bm\theta))\},
\end{eqnarray}
and hence the maximum likelihood (ML) estimator $\hat{\bm\theta}_{\rm ML}=(\hat{\bm\alpha}_{\rm ML},\hat{\bm\psi}_{\rm ML})$ is obtained by the maximization of the log-likelihood function $L(\bm\theta)$ in $\bm\theta=(\bm\alpha,\bm\psi)$.
The estimated policy function for a given covariate $\bm x$ is provided  by 
\begin{eqnarray}\label{ml-policy}
\hat{\cal D}_{\rm ML}(\bm x)=\argmax_{a\in{\cal A}}g(\bm x,a,\hat{\bm\psi}_{\rm ML})
\end{eqnarray}
via plugging in the ML estimator $\hat{\bm\psi}_{\rm ML}$ of the policy parameter $\bm\psi$.
The estimating function is given by
\begin{eqnarray}\nonumber 
{\cal E}_{\rm ML}(\bm\alpha,\bm\psi)=\{Y-Q(\bm X,A,\bm\alpha,\bm\psi)\}
\begin{bmatrix}
\frac{\partial}{\partial\bm\alpha} \nu (\bm X ,A ,\bm\theta)\\
\frac{\partial}{\partial\bm\psi} \nu (\bm X ,A ,\bm\theta)
\end{bmatrix},
\end{eqnarray}
which is unbiased, that is $\mathbb E[{\cal E}_{\rm ML}(\bm\theta)]={\bf0}$ because $\mathbb E[Y|\bm X,A]=Q(\bm X,A,\bm\alpha,\bm\psi)$.
Thus, $\hat{\bm\theta}_{\rm ML}$ is consistent for $\bm\theta$.
The usual regression analysis  is conducted in the generalized linear model, and  the ML-based policy estimator  ${\cal D}_{\rm ML}(\bm x)$ as in \eqref{ml-policy} is directly yielded in the single stage case.
This is called the Q-learning, which is extended to a multi-stage case by a stagewise regression analysis with a pseudo outcome in backward learning.  
Wallace \& Moody (2015) discusses the balanced weighting for ${\cal E}_{\rm ML}$ 
in the case of normal regression, which gives doubly robust properties of the weighted ML
estimator.

Let us consider  the information divergence defined on ${\cal Q}$.
For example, the extended Kullback-Leibler (eKL) divergence is given by
\begin{eqnarray}\nonumber 
D_{\rm eKL}(Q_0,Q_1)=\int_{\cal X}\sum_{a\in{\cal A}}\Big\{Q_0(\bm x,a)\log \frac{Q_0(\bm x,a)}{Q_1(\bm x,a)}-{Q_1(\bm x,a)}+{Q_0(\bm x,a)}\Big\}p(\bm x)d\bm x,
\end{eqnarray}
which is equal to the $\beta$-power divergence taking the limit of $\beta$ to $0$.

Apparently, the KL divergence in $\cal P$ and $D_{\rm eKL}(Q_0,Q_1)$ are unrelated.
However, if $Y$ has a Poisson distribution, that is
\begin{eqnarray}\nonumber 
 p(y|\bm x, a)= \frac{Q(\bm x,a)^y}{y!}\exp\big\{- Q(\bm x,a)\big\}
\end{eqnarray}
for $y=0,1,...$, then $$D_{\rm eKL}(Q_0,Q_1)=\mathbb E[D_{\rm KL}( p_{1}(\cdot|\bm X,A),p_{ 2}(\cdot|\bm X,A))].$$
Thus, there is a close relationship between the information divergences defined on the space of Q-functions and the space of
 conditional density functions of the outcome.

We have discussed a direct modeling of the Q-function as in the semiparametric model $\cal M$ in \eqref{Q-model} and the parametric model ${\cal M}^*$ in \eqref{Q-model1}.
Under these models, the $\gamma$-power and $\beta$-power MDEs are discussed from the viewpoint of
model consistency.
The semiparametric consistency of the $\gamma$-power MHD is observed in Theorem \ref{thm5}. 
Alternatively, in this section, the standard discussion by the maximum likelihood is given, and pointed out the close relationship between the two approaches.
However, the maximum likelihood estimator $\hat{\bm\psi}_{\rm ML}$ does not in general satisfy the semiparametric consistency, which is similar to the case of  the $\beta$-power MDE.

\section{Simulation study}\label{sec7}

We have a brief simulation study for checking numerical behaviors of the $\gamma$-power MDE and $\beta$-power MDE in  parametric and semiparametric situations, cf Tsiatis \cite{tsiatis2006}. 

Let $X$ be generated from a normal distribution $N (1,0.5)$, and $A$ from a trinomial distribution with
a cell probability vector $(1/3,1/3,1/3)$.
The outcome $Y$ is generated from an exponential distribution with the conditional mean
\begin{eqnarray}\nonumber 
Q(x,a,\bm \alpha,\bm\psi)=\exp\{\alpha_1 x+\alpha_0\}\exp\{x \psi_1 a+\psi_0 a\},
\end{eqnarray}
where $\bm \alpha$ and $\bm\psi$ are fixed as $(-1.0,-2.0)$ and $( 2.0,-1.0)$, respectively.
Then, the $\gamma$-power loss function is given by
\begin{eqnarray}\nonumber 
L_\gamma(\bm\psi)=\frac{1}{\gamma}\frac{1}{n}\sum_{i=1}^n  \frac{ Y_i \exp\{\gamma(X_i \psi_1 A_i+\psi_0 A_i)\}}{\big[\sum_{a=1}^3 \exp\{(\gamma+1)(X_i \psi_1 a +\psi_0 a)\}\big]^{\frac{\gamma}{\gamma+1}}}
.
\end{eqnarray}
On the other hand, the $\beta$-power loss function is given by
\begin{eqnarray}\nonumber 
L_\beta(\bm\alpha, \bm\psi)&=&\frac{1}{n}\sum_{i=1}^n \bigg[\frac{ \sum_{a=1}^3 \exp\{(\beta+1)(\alpha_1 X_i+\alpha_0+X_i \psi_1 a +\psi_0 a)\} }{\beta+1}\nonumber\\
&&\hspace{10mm}- \frac{Y_i \exp\{\beta(\alpha_1 X_i+\alpha_0+X_i \psi_1 A_i+\psi_0 A_i\}}{\beta} \bigg].
\end{eqnarray}

We simulated the random sequence with $n=500$ and 300 replications, where the true value  of the policy parameter is $(\psi_1,\psi_0)=(2.0,-1.0)$.
For a fair comparison, 
we set the pair of power indices as $\beta=\gamma=-1.5$.
In this specific case, both estimators with positive power indices relatively tend to have less performance than those with negative indices.  
We observe that there was no superiority or inferiority relationship between the $\beta$-power and $\gamma$-power MDEs, cf. Table \ref{table1}(a).
In conclusion, both estimators are confirmed the consistency for $(\psi_1,\psi_0)$ since the mean vectors for 300 replications are nearly equal to the true $(2.0,-1.0)$ with reasonable Root Mean Squares Estimates (RMSEs).
\begin{table}[hbtp]
  \caption{Performance of  the $\gamma$-power and $\beta$-power MDEs}
  \label{table1}
  \centering
\vspace{2mm}
(a). The corrected model case

\vspace{2mm}
  \begin{tabular}{ccc}
  \hline 

    Method  & Mean of $(\hat\psi_1,\hat\psi_0)$ &  RMSE of  $(\hat\psi_1,\hat\psi_0)$
 \\
    \hline \hline

    $\gamma$-power& $(2.012, -1.009)$  & $(0.251, 0.342)$\\
    $\beta$-power  & $(2.074, -1.005)$  & $(0.569, 0.259)$ \\
    \hline
  \end{tabular}

\vspace{3mm}
(b). The misspecified model case

\vspace{2mm}
  \begin{tabular}{ccc}
    \hline
    Method  & Mean of  $(\hat\psi_1,\hat\psi_0)$ &  RMSE of $(\hat\psi_1,\hat\psi_0)$  \\
    \hline \hline

    $\gamma$-power& $(2.029, -1.018)$  & $(0.273, 0.373)$\\
    $\beta$-power  & $(\bf{3.064}, -\bf{0.428})$  & $(\bf{3.069}, \bf{2.290})$ \\
    \hline
  \end{tabular}
\end{table}

We next change the simulation scenario.
The conditional mean of $Y$  is misspecified as
\begin{eqnarray}\label{mis}
Q(x,a,\bm\psi)=\exp\{-(x-1)^2\}\exp\{x \psi_1 a+\psi_0 a\}.
\end{eqnarray}
Here $\bm\psi$ is still fixed as $(2.0,-1.0)$, however, the working model $\exp\{\alpha_1 x+\alpha_0\}$
is replaced with $\exp\{-(x-1)^2\}$.
Similarly, we have a numerical comparison between the $\beta$-power and $\gamma$-power MDEs fixing the same setting as above
except for \eqref{mis}, cf. Table \ref{table1}(b).
The mean of the  $\beta$-power MDE had a larger bias and
the RMSEs  became larger because of  the bias arisen by the misspecification of the nuisance parameter $\bm \alpha$.
We conclude that the $\gamma$-power MDE is much more efficient than  the $\beta$-power MDE.
In this way,  we confirm that $\gamma$-power MDE has robustness for misspecification.

This simulation experiment does not sufficiently present the comparison conclusions.
An adapted selection for the power indices $\gamma$ and $\beta$ based on the dataset can be given by $K$-fold cross-validation or information criterion methods, however, we just fixed the same value.
We need to investigate the distributional behaviors of these estimators given in a more practical setting with a  multi-stage setting.  
However, we briefly confirmed the result of the semiparametric consistency in Theorem \ref{thm5}.

\section{Discussion}\label{sec8}

We have discussed the $\gamma$-power divergence 
on the Q-function space and the $\gamma$-MDE from statistical perspectives.
Theorem \ref{thm1} characterizes $\gamma$-power divergence 
based on the policy equivalence relation defined on the Q-function space, in which the $\gamma$-power divergence 
is invariant under the policy equivalence relation as in Corollary \ref{cor1}. 
Specific forms of $\gamma$-power divergence are discussed taking some values of the power index $\gamma$ 
in Theorems \ref{thm2}, \ref{thm3} and \ref{thm4}.
As a main result, we show the semiparametric consistency of the $\gamma$-power MDE for the multiplicative semiparametric model in Theorem \ref{thm5}.
In effect, this advantageous property is valid for any power index $\gamma$ of $\mathbb R$.
We do not have a good understanding of the choice for $\gamma$ at the present stage, in which an empirical study similar to the simulation in Section \ref{sec7} suggests giving smaller RMSEs when $\gamma$ is negative.
A data-adapted selection for $\gamma$ should be formulated in a near future discussion.

The result for the semiparametric consistency of the $\gamma$-power MDE is closely related to the theory of the partial likelihood for the proportional hazard model in survival analysis, in which the likelihood is factorized  
into the partial and remainder likelihood functions, cf Cox \cite{cox1975} and Battey et al. \cite{Battey-etal2022}.
The semiparametric models are both multiplicative, and hence  the $\gamma$-power MDE  for the survival analysis can be applied to with more discussion for a probabilistic mechanism of the time-to-event outcomes.


Finally, general remarks for information geometry are given in a personal view.
Information geometry is still a developing discipline, in which there potentially remain a vast of problems to be clarified by information geometry.
Obviously, physics has established and expanded many magnificent ideas in the physical space from a quantum scale to the universe scale;
mathematics presents elegant abstraction to strengthen such physical ideas.
However, we have to recognize that such an expansion should be done in the human brain system, in which virtual reality and data-driven societies are
materialized in the cyberspace in virtue of the progress of information science and technology, cf. \cite{harari2016}.
If there does not exist any consciousness, all the physical phenomena are sensibly observed or not?     
Any fundamental solution for such a basic problem regarding consciousness and phenomena has not yet been fulfilled.
Information geometry may be able to challenge to provide recognizable solutions via combining ideas such as information, entropy, and divergence
in integration from physical space to information space.


\end{document}